\documentclass[preprint,12pt,authoryear]{elsarticle}

\usepackage{graphicx}%
\usepackage{multirow}%
\usepackage{amsmath,amssymb,amsfonts}%
\usepackage{amsthm}%
\usepackage{mathrsfs}%
\usepackage[title]{appendix}%
\usepackage{xcolor}%
\usepackage{textcomp}%
\usepackage{manyfoot}%
\usepackage{booktabs}%
\usepackage{listings}%
\usepackage{subcaption} %
\usepackage[table,xcdraw]{xcolor}
\usepackage{graphicx}
\usepackage{lscape}
\usepackage{graphicx}
\usepackage{float}
\usepackage{tablefootnote}
\usepackage[ruled,vlined,linesnumbered]{algorithm2e}
\newcounter{AlgoLineStore}

\usepackage{amssymb}
\usepackage{amsmath}

\theoremstyle{thmstyleone}%
\newtheorem{theorem}{Theorem}
\newtheorem{proposition}{Proposition}%
\newtheorem{corollary}{Corollary}%

\theoremstyle{thmstyletwo}%
\newtheorem{lemma}{Lemma}%
\newtheorem{result}{Result}%

\theoremstyle{thmstylethree}%
\newtheorem{definition}{Definition}%

\journal{Not Sure}
\usepackage[colorlinks=true, linkcolor=blue, citecolor=blue, urlcolor=blue]{hyperref}

\usepackage{etoolbox}

\makeatletter
\patchcmd{\pprintMaketitle}{\hrule\vskip12pt}{}{}{}
\patchcmd{\pprintMaketitle}{\hrule\vskip12pt}{}{}{}
\patchcmd{\MaketitleBox}{\hrule\vskip12pt}{}{}{}
\patchcmd{\MaketitleBox}{\hrule\vskip12pt}{}{}{}
\makeatother

\begin{document}

\begin{frontmatter}


\title{Equilibria in heterogeneous multi-strategy network games} 


\author[1]{Shan Pei\fnref{equal}}\ead{peishan113@163.com}

\author[2,3]{Wenjie Cao\fnref{equal}}\ead{wenjiecao@mail.bnu.edu.cn}

 

\author[3]{Boyu Zhang\corref{cor1}}\ead{zhangby@bnu.edu.cn}

\fntext[equal]{Shan Pei and Wenjie Cao contributed equally to this work.}
\cortext[cor1]{Corresponding author.}

\affiliation[1]{organization={College of Science, China Agricultural University},
            city={Beijing},
            postcode={100083}, 
            country={PR China}}

\affiliation[2]{organization={Power Grid Planning and Research Center of Guizhou Power Grid Co., Ltd},
           city={Guiyang},
           state={Guizhou},
           postcode={550000}, 
           country={PR China}}

\affiliation[3]{organization={Laboratory of Mathematics and Complex Systems, Ministry of Education, School of Mathematical Sciences, Beijing Normal University},
            city={Beijing},
            postcode={100875}, 
            country={PR China}}


\begin{abstract}
In this paper, we analyze a multi-strategy network game with three types of players, conformists, rebels, and stubborn agents. Conformists adopt the strategy that is most common among their neighbors, rebels adopt the least common, and stubborn agents adhere to a fixed strategy. We study the existence and structure of pure strategy Nash equilibrium (PNE). On arbitrary networks, we establish sufficient conditions for PNE existence, and prove that in a broad class of large random networks the probability of PNE existence tends to zero as the network size grows. For several specific network architectures, such as complete network, lines, rings, trees, and stars, we derive necessary and sufficient conditions or develop an algorithm for determining PNE existence and characterize the equilibrium strategy frequencies. Collectively, these results offer a unified perspective that PNE is likely to exist when every conformist has more conformist and stubborn neighbors, and fails when the network game has numerous conformist-rebel edges.
\end{abstract}




\end{frontmatter}


\newpage
\section{Introduction}


Over the last decades, the theory of network games has grown rapidly. The two most extensively studied classes are games of strategic complements and games of strategic substitutes. In the former, a player’s incentive to adopt an action rises with the number of neighbors taking that action, generating coordination motives; in the latter, that incentive falls, generating anti-coordination motives. Most existing analyses treat these two forms of interaction separately and assume that all players share the same qualitative preference type \citep{Galeotti2010,Hernandez2013,Bramoulle2014,Jackson2015,Bramoulle2016,Ramazi2016}.

In real-world systems, however, players often differ in their strategic preferences. Some tend to coordinate with neighbors, others to anti-coordinate, and still others adhere to fixed strategies. For example, in social networks, some individuals follow peer pressure, others rebel against it, and opinion leaders or committed activists remain steadfast. Such heterogeneity is common in fashion \citep{Jackson2008, Zhang2018}, the diffusion of innovations \citep{Jackson2007, Young2011}, and political participation \citep{Acemoglu2013,Yildiz2013,Stewart2019}. When conformists, rebels, and stubborn agents interact, the equilibrium problem becomes more complex: a strategy that suits one type may conflict with the incentives of its neighbors, and fixed players can either stabilize or intensify these conflicts. Equilibrium feasibility thus depends not only on the frequencies of each type, but also on their network locations, fixed actions, and the number of available strategies.

Studies that allow such heterogeneity have largely concentrated on binary-choice network games with two player types. For instance, \citet{Zhang2018} and \citet{Cao2019} examine the fashion game featuring conformists and rebels, linking equilibrium existence to homophily. Other work studies conformists together with stubborn agents \citep{Cao2024} or the effect of committed individuals on binary opinion dynamics \citep{Acemoglu2013,Yildiz2013,Stewart2019}. More recent contributions extend the analysis to networks with two arbitrary types \citep{Pei2024} or heterogeneous populations of coordinators and anti-coordinators \citep{Ramazi2016, Ramazi2023, Rajaee2024}. These studies characterize equilibrium conditions, convergence, stochastic stability, and persistent oscillations, showing that conflicting behavioral rules can generate dynamics absent in homogeneous populations. An important exception is \citet{Cao2026}, who consider binary-choice network games with fully heterogeneous payoffs and show that they can be recast as equivalent games with conformists, rebels, and stubborn agents. However, their analysis remains limited to two strategies. A systematic study of Nash equilibrium in network games with multiple behavioral types and strategies is still lacking.

In this paper, we study an $m$-strategy network game with three types of players, namely conformists, rebels, and stubborn agents, which we refer to as a CRS game. Conformists prefer to match the strategy of their neighbors, rebels prefer to differ, and stubborn agents remain fixed at a predetermined strategy. Our aim is to understand how coordination incentives, anti-coordination incentives, and fixed behavior jointly determine equilibrium outcomes in a multi‑strategy setting. We derive general sufficient conditions for the existence of pure‑strategy Nash equilibria (PNE) and identify local configurations that preclude equilibrium. For complete networks, lines, rings, trees, and stars, we further obtain necessary and sufficient conditions or efficient constructive procedures. In doing so, we offer a unified framework for understanding how complementarity, substitutability, exogenous constraints, and network topology jointly shape equilibrium in multi‑strategy network games.

We begin with benchmark results for network games involving only two player types, including a characterization of PNE in conformist-rebel games on bipartite networks. Turning to the general CRS game on arbitrary networks, we derive sufficient conditions for the existence of a PNE. To characterize non-existence, we identify a simple local configuration whose presence precludes any PNE in the entire network. Using this local obstruction, we prove that a PNE fails with probability tending to one in a broad class of large random networks. These results indicate that, while equilibrium can be sustained by particular network structures, the heterogeneity of behavioral types increasingly generates incompatible incentives as the random network grows.

We next derive necessary and sufficient conditions for PNE existence on several specific network architectures, including complete networks, lines, rings, trees, and stars. On complete networks, we fully characterize the equilibrium structure, showing that conformists coordinate on the strictly most common strategy while rebels distribute themselves as evenly as possible across strategies. For lines and rings with at least three strategies, a PNE exists if and only if every conformist has at least one conformist or stubborn neighbor. The same condition extends to trees when the number of strategies exceeds the maximum degree. We also provide an algorithm that determines whether a PNE exists on a finite tree and constructs one if it does. For star networks, we derive a sharper characterization based on the type of central player and the composition of leaf nodes, which applies even when the number of strategies does not exceed the maximum degree.

Overall, our results show how network topology and the distribution of types jointly determine equilibrium outcomes.  Conformists need sufficient local support to coordinate, while rebels need strategic flexibility to avoid matching. A PNE emerges when these opposing forces can be simultaneously accommodated, which requires every conformist to have enough conformist or stubborn neighbors, and fails when numerous conformist-rebel edges create unavoidable local conflicts. Our analysis thus provides a systematic framework for studying coordination and conflict in heterogeneous populations with multiple strategies.

The remainder of the paper is organized as follows. Section 2 introduces the $m$-strategy CRS game models. Section 3 analyzes the existence of PNE for general network structures. Section 4 characterizes the equilibrium structure on several specific network structures. Section 5 concludes and discusses extensions.

\section{Model}
In this paper, we consider games played on fixed networks with three types of players, namely conformists, rebels, and stubborn agents (CRS game for short). A conformist prefers to take the most common strategy in his/her neighbors, a rebel prefers to take the least common strategy, and a stubborn agent never changes his/her strategy. Formally, a CRS game can be represented by a system $G = (C, R, S, E, A, U)$, where $C$, $R$, and $S$ are the sets of conformists, rebels, and stubborn agents, respectively, $E$ is the set of undirected edges, $A=\left\{a_1, \ldots, a_m\right\}$ is the set of strategies, and $U$ is the set of payoff functions.

\par \textbullet \ 
$N=C \cup R \cup S$ is the set of players, also referred to as nodes of a network. Suppose that there are $n$ players in total, and the numbers of conformists, rebels, and stubborn agents with strategy $a_i$ are $c$, $r$, and $s_i$, respectively. Thus, the sets for the three types of players can be represented as $C=\{1, \ldots, c\}, R=\{c+1, \ldots, \mathrm{c}+r\}$, and $S=\{c+r+1, \ldots, c+r+s_1, \ldots,  c+r+s\}$, where $s=\sum_{i=1}^m s_i$ and $c+r+s=n$.

\par \textbullet \ 
$E \subseteq N \times N$ is the set of undirected edges. We assume that $E \neq \emptyset$ and $i i \notin E$ for all $i \in N$. Player $i \in N$ can observe the strategy of player $j \in N$ (i.e., $j$ is a neighbor of $i$) if and only if $i j \in E$. For all $i \in N$, we use $N_i(G)=\{j \mid i j \in E\}$ to denote the set of $i$'s neighbors in $G$ and $d_i:=\left|N_i\right|$ to denote the number of $i$'s neighbors. 

\par \textbullet \ 
$A=\left\{a_1, \ldots, a_m\right\}$ is the strategy set for each player. We use $A^N$ to denote the set of pure strategy profiles. Given a strategy profile $\boldsymbol{x}=\left(x_1, \ldots, x_c, x_{c+1}, \ldots, x_{c+r}, x_{c+r+1}, \ldots, x_n\right) \in A^N$, $x_i$ denotes the strategy of player $i$, and $\boldsymbol{x}_{-i}$ denotes the profile of the strategies of players other than $i$. In addition, we do not consider mixed strategies throughout this paper.

\par \textbullet \ 
$U=\left(u_i(\cdot)\right)_{i \in N}$ is the collection of payoff functions. Following \citet{Cao2024} and \citet{Cao2026},  payoff functions for the three types of player can be written as
\begin{equation}
u_i(\boldsymbol{x})=\left\{\begin{array}{cl}
\left|\left\{j \in N_i: x_j=x_i\right\}\right|, & \text { if } i \in C \\
\left|\left\{j \in N_i: x_j \neq x_i\right\}\right|, & \text { if } i \in R \\
0, & \text { if } i \in S
\end{array}\right.  \ \ \ .
\label{eq1}
\end{equation}
Thus, the payoff of a conformist is taken as the number of neighbors using the same strategy as him/her, and the payoff of a rebel is taken as the number of neighbors using different strategies. 

Following the notations above, a pure strategy Nash equilibrium (PNE) of the CRS game is defined as follows.

\begin{definition}
A strategy profile $\boldsymbol{x}^* \in A^N$ is a PNE if and only if $u_i\left(\boldsymbol{x}^*\right) \geq u_i\left(x_i, \boldsymbol{x}_{-i}{ }^*\right)$ for $\forall i \in N$ and $\forall x_i \in A$, where $u_i(\boldsymbol{x})$ is defined in Eq.(\ref{eq1}).
\label{def1}
\end{definition}

From Definition \ref{def1}, at a PNE $\boldsymbol{x}^*$, $x_i^*$ is (one of) the most common strategy among the neighbors of player $i$ for all $i \in C$, and $x_j^*$ is (one of) the least common strategy among the neighbors of player $j$ for all $j \in R$.

\section{Equilibrium analysis: General network structures}

\subsection{Benchmark cases}

Previous studies have focused on the existence of PNE in network games with two types of players, such as games combining conformists with stubborn agents or with rebels. Some of these results can be applied or extended to our model. In this section, we briefly summarize these results.

\begin{result}
For any network, the $m$-strategy game with conformists and stubborn agents (CS game) has at least one PNE.
\label{re1}
\end{result}

This result was introduced in \citet{Cao2024}, showing that the $m$-strategy CS game is an exact potential game \citep{Monderer1996}. A similar result can be established for the $m$-strategy game with rebels and stubborn agents (RS game) using the potential game method.

\begin{result}
For any network, the $m$-strategy game with rebels and stubborn agents (RS game) has at least one PNE.
\label{re2}
\end{result}

\begin{proof}
For an RS game, we define the potential function $\phi:A^N\to\mathbb{R}$ as follows:
$$
\phi(\boldsymbol{x}):=\frac{1}{2}(\sum_{k \in R}\left|\left\{j \in N_k: x_j \neq x_k\right\}\right|+\sum_{k \in R}\left|\left\{j \in N_k \cap S: x_j \neq x_k\right\}\right|).
$$
Since stubborn agents do not change their strategies, we only need to focus on the changes of strategies and utilities of rebels. For $\forall i \in R$, $\forall \boldsymbol{x}_{-i} \in A^N_{-i}$, and $\forall x_i, x_i^{\prime} \in A$, it is easy to check 
$$
\begin{aligned}
&\phi\left(x_i, \boldsymbol{x}_{-i}\right)-\phi\left(x_i^{\prime}, \boldsymbol{x}_{-i}\right) \\
=&\left|\left\{j \in N_i: x_j \neq x_i\right\}\right|-\left|\left\{j \in N_i: x_j \neq x_i^{\prime}\right\}\right| \\
=&u_i\left(x_i, \boldsymbol{x}_{-i}\right)-u_i\left(x_i', \boldsymbol{x}_{-i}\right).
\end{aligned}
$$
This means that the RS game is an exact potential game and admits at least one PNE.
\end{proof}

The situation of games with conformists and rebels (CR game) is much more complicated. A CR game may not have any PNE, as illustrated by the two-player matching pennies game. \citet{Zhang2018} analyzed the properties of PNE for the 2-strategy CR game (called the fashion game in their paper). They 
show that a PNE exists if the network exhibits strong conformist homophily or strong rebel homophily, i.e., every conformist has sufficiently many conformist neighbors, or every rebel has sufficiently many rebel neighbors. Conversely, on bipartite networks whose two vertex classes are $C$ and $R$ (i.e., contain no CC or RR edges), a PNE exists only if every player has even degree. Here we extend the result of bipartite networks to the $m$-strategy setting.

\begin{result}
On a bipartite network whose two vertex classes are $C$ and $R$, an $m$-strategy CR game admits a PNE only if $m \mid d_i$ for all $i\in C\cup R$. Furthermore, at any PNE, the neighbors of every agent are evenly distributed across the $m$ strategies.
\label{re3}
\end{result}

\begin{proof}
We show that, on a bipartite CR network whose two vertex classes are $C$ and $R$, a strategy profile is a PNE if and only if the neighbors of every vertex are evenly distributed across the $m$ strategies. For any strategy profile $\mathbf{x}$, any vertex $i$, and any strategy $a\in A$, define $n_i(a,\mathbf{x}) := |\{j\in N_i:x_j=a\}|.$ Thus, $n_i(a,\mathbf{x})$ is the number of neighbors of $i$ choosing strategy $a$. Also define $M(\mathbf{x}) := |\{\{i,j\}\in E:x_i=x_j\}|$, i.e., $M(\mathbf{x})$ is the total number of edges whose two endpoints choose the same strategy.

Suppose first that $\mathbf{x}^*$ is a PNE. For each conformist $i\in C$,
let $d_i=|N_i|$. Since $i$ is a conformist, it must choose a strategy that is
most frequent among its neighbors. Hence, $n_i(x_i^*,\mathbf{x}^*) = \max_{a\in A} n_i(a,\mathbf{x}^*) \geq \frac{d_i}{m}$. Summing over all conformists gives $M(\mathbf{x}^*) = \sum_{i\in C} n_i(x_i^*,\mathbf{x}^*) \geq \sum_{i\in C}\frac{d_i}{m} = \frac{|E|}{m}.$ On the other hand, for each rebel $j\in R$, let $d_j=|N_j|$. Since $j$ is a rebel, it must choose a strategy that is least frequent among its neighbors. Hence, $n_j(x_j^*,\mathbf{x}^*) = \min_{a\in A} n_j(a,\mathbf{x}^*) \leq \frac{d_j}{m}.$ Summing all rebels gives $M(\mathbf{x}^*) = \sum_{j\in R} n_j(x_j^*,\mathbf{x}^*) \leq \sum_{j\in R}\frac{d_j}{m} = \frac{|E|}{m}.$ Therefore, $M(\mathbf{x}^*)=\frac{|E|}{m}.$

It follows that all above inequalities must hold with equality. In particular, for every conformist $i\in C$, $\max_{a\in A} n_i(a,\mathbf{x}^*)=\frac{d_i}{m}$. Since $\sum_{a\in A} n_i(a,\mathbf{x}^*)=d_i$, this implies $n_i(a,\mathbf{x}^*)=\frac{d_i}{m}$ for all $a\in A$. Similarly, for every rebel $j\in R$, $\min_{a\in A} n_j(a,\mathbf{x}^*)=\frac{d_j}{m}$, and therefore $n_j(a,\mathbf{x}^*)=\frac{d_j}{m}$ for all $a\in A$. Thus, at any PNE, the neighbors of every agent must be evenly distributed across the $m$ strategies. In particular, since each $n_i(a,\mathbf{x}^*)$ is an integer, the equality $n_i(a,\mathbf{x}^*)=\frac{d_i}{m}$ implies that $m\mid d_i$ for all $i\in C\cup R$. \end{proof}

We further note that any strategy profile $\mathbf{x}$ in which every vertex has equally many neighbors playing each strategy is a PNE. In such a profile, each conformist observes all strategies equally often among its neighbors, making its own strategy a locally most frequent one. Similarly, each rebel observes all strategies equally often, making its own strategy a locally least frequent one. Hence, every player plays a best response, and $\mathbf{x}$ is a PNE. The next result shows that the PNE existence condition in Result~\ref{re3} is also sufficient when the bipartite network is complete.

\begin{result}
On a complete bipartite network whose two vertex classes are $C$ and $R$, an $m$-strategy CR game admits a PNE if and only if $m \mid c$ and $m \mid r$. Furthermore, at any PNE, conformists and rebels are each evenly distributed across the $m$ strategies.
\label{re4}
\end{result}

\begin{proof}
By Result~\ref{re3}, a PNE exists if and only if there is a strategy profile  in which every vertex has the same number of neighbors choosing each strategy. In the complete bipartite network, every conformist observes all rebels and every rebel observes all conformists. Thus this local balance condition is equivalent to the global balance condition that rebels and conformists are each evenly distributed across the $m$ strategies, namely $m\mid r$ and $m\mid c$. The final statement also follows from Result~\ref{re3}.
\end{proof}




\subsection{General cases}

In this subsection, we assume that the $m$-strategy CRS game includes all three types of player, conformists, rebels, and stubborn agents, and we study the properties of PNE for this game on arbitrary networks. 

\begin{theorem}\label{the1}
The $m$-strategy networked CRS game admits a PNE if one of the following conditions holds: (i) the game does not have conformist-rebel (CR) edges;
(ii) there exists a strategy $a\in A$ such that each conformist has at least $\frac{1}{2}$ of his/her neighbors in the set $C \cup S_a$. 
\end{theorem}
\begin{proof}
We verify the existence of a PNE under each condition. 

For condition (i), a CRS game without CR edges can be decomposed into a CS game and an RS game. Result 1 and Result 2 show that both the CS and RS games are exact potential games and therefore possess a PNE.

For condition (ii), assigning strategy $a$ to all conformists guarantees that no conformist has the incentive to deviate. Once their strategies are fixed in this way, the game is reduced to an RS game and admits a PNE.
\end{proof}

The above theorem provides two sufficient conditions for the existence of PNE in $m$-strategy CRS games. These conditions are a direct extension for Theorem 2 in \citet{Cao2026}, which provides sufficient conditions for the existence of PNE for $2$-strategy CRS games.

As for the question of nonexistence of a PNE, to the best of our knowledge there is no general method to determine this for a CRS game on an arbitrarily given network. We first identify a simple but critical local structure whose existence precludes the game from admitting a PNE.

\begin{definition}
\label{def:M-star-degree}
A pair of nodes $(u,v)$ is said to form a local structure $M^*$ if there exist integers $k_u,k_v\ge 0$ such that node $u \in C$ has one rebel neighbor $v$ and $mk_u$ stubborn neighbors, where exactly $k_u$ of whom adopt each strategy $a\in A$, and likewise, node $v \in R$ has one conformist neighbor $u$ and $mk_v$ stubborn neighbors, where exactly $k_v$ of whom adopt each strategy $a\in A$.
\end{definition}

\begin{lemma}\label{lem1}
If a CRS game has (at least) a pair of nodes $(u,v)$ satisfying the local structure $M^*$, then the game does not admit a PNE. 
\end{lemma}

\begin{proof}
We prove the result by contradiction. Suppose that the CRS game has a pair of nodes $(u,v)$ satisfying the local structure $M^*$. By definition, node $u$ is a conformist and node $v$ is a rebel. Furthermore, except for the edge connecting $u$ and $v$, all neighbors of $u$ and $v$ are stubborn players uniformly distributed over the strategy set $A$. 

Suppose that the game has a PNE, denoted by $\mathbf{x}^*$. If $x_u^*\neq x_v^*$, then the payoff of node $u$ is $u_u(\boldsymbol{x}^*)=k_u$, and its payoff can increase by 1 by switching to strategy $x_v^*$. On the other hand, if $x_u^*=x_v^*$, then the payoff of node $v$ is $u_v(\boldsymbol{x}^*)=(m-1)k_v$, and its payoff can increase by 1 by switching to any different strategy. Thus, it is impossible for both $u$ and $v$ to play the best responses simultaneously, and therefore $\mathbf{x}^*$ cannot be a PNE.
\end{proof}

We next prove the non-existence of PNE in large random networks by showing that the probability of lacking a local structure $M^*$ vanishes as the network size grows.


\begin{theorem}\label{the2}
Consider a network CRS game on a random graph with $n$ nodes, degree distribution $\mathbf{d}$, and type distribution $\mathbf{\pi}=(\pi_C,\pi_R,\pi_{S_1},\dots,\pi_{S_m})$. We assume
(i) $\sup_n \frac1n\sum_i d_i^2 < \infty$,
(ii) there exist constants $\alpha>0$ and $K>0$ such that the induced subgraph on $V_n^*=\{i: d_i\equiv1\pmod m,\ d_i\le K\}$ contains at least $\alpha n$ edges,
(iii) conditional on the graph, player types are drawn independently according to $\mathbf{\pi}$ and every type has positive probability.
Then, the probability that the game admits a PNE tends to zero as $n\to \infty$.
\end{theorem}

\begin{proof}
Write $\pi_{\min}= \min\{\pi_C,\pi_R,\pi_{S_1},\dots,\pi_{S_m}\} >0$. Let $E_n$ be the set of edges whose both endpoints lie in $V_n^*$. By (ii), $|E_n|\ge \alpha n$ for large $n$. For $(u,v)\in E_n$, set $k_u=(d_u-1)/m$ and $k_v=(d_v-1)/m$, which are integers by definition of $V_n^*$. Define the event $A_{(u,v)}$ that
(a) $u$ is a conformist and $v$ is a rebel, (b) among the $m k_u$ neighbors of $u$ other than $v$, each stubborn type appears exactly $k_u$ times, (c) among the $m k_v$ neighbors of $v$ other than $u$, each stubborn type appears exactly $k_v$ times.
Let $Y_{(u,v)}=\mathbf 1_{A_{(u,v)}}$ and $Y_n=\sum_{(u,v)\in E_n}Y_{(u,v)}$.

Let $S_{(u,v)}=\{u,v\}\cup N(u)\cup N(v)$. Because $d_u,d_v\le K$, $|S_{(u,v)}|\le 2K+2$. A type assignment on $S_{(u,v)}$ that fulfills $A_{(u,v)}$ exists (common neighbours can be assigned types first without exceeding $\min\{k_u,k_v\}$ of each, and remaining neighbours filled to reach exactly $k_u,k_v$). Under independent assignment, the probability of this particular configuration is at least $\pi_C\pi_R\,\pi_{\min}^{|S_{(u,v)}\setminus\{u,v\}|}\ge \pi_C\pi_R\,\pi_{\min}^{2K}:=\rho$. Hence $\mathbb{E}[Y_{(u,v)}]\ge\rho$ and $\mathbb{E}[Y_n]\ge \rho|E_n|\ge \rho\alpha n$.

The random variable $Y_{(u,v)}$ depends only on types in $S_{(u,v)}$. Thus if $S_{(u,v)}\cap S_{(u',v')}=\varnothing$, the indicators are independent. Only pairs with overlapping neighborhoods contribute to the covariance. For a vertex $w$, let $N_w$ be the number of edges $(u,v)\in E_n$ with $w\in S_{(u,v)}$. Since $u,v\in V_n^*$ have degree at most $K$, $N_w\le \sum_{z\in N(w)\cap V_n^*}d_z + d_w\mathbf 1_{\{w\in V_n^*\}} \le K(d_w+1)$. Consequently, the number of ordered pairs of edges with overlapping neighborhoods is at most
\[
\sum_w N_w^2 \le K^2\sum_w (d_w+1)^2 \le 2K^2\sum_w (d_w^2+1) = O(n),
\]
using the bounded second‑moment assumption. As $|\operatorname{Cov}(Y_{(u,v)},Y_{(u',v')})|\le1$, we obtain $\operatorname{Var}(Y_n)=O(n)$.

Chebyshev's inequality yields
$$
\mathbb{P}(Y_n=0)\le \mathbb{P}\bigl(|Y_n-\mathbb{E}[Y_n]|\ge \mathbb{E}[Y_n]\bigr)
\le \frac{\operatorname{Var}(Y_n)}{(\mathbb{E}[Y_n])^2} \le \frac{C' n}{(\rho\alpha n)^2}=O(n^{-1})\to 0.
$$
Thus, $\mathbb{P}(Y_n\ge 1)\to 1$. By Lemma~\ref{lem1}, the existence of such a local configuration precludes a PNE. Therefore, the game admits no PNE with probability tending to one as $n \to \infty$.

\end{proof}


Theorem~\ref{the2} implies that PNE is absent with high probability in large random networks for a broad class of degree distributions, such as Poisson, exponential, binomial, and lognormal distributions, provided that the induced subgraph on $V_n^*$ contains sufficiently many edges for the local structure $M^*$ to appear with probability tending to one as the network grows\footnote{For all these network models, the limiting degree distribution has finite second moment. Let $p_k$ denotes the degree distribution with mean $\mu$. Then, for sufficiently large $n$ and $K$, $\frac{|E(V_n^*)|}{n} \to \frac{1}{2\mu} \left( \sum_{\substack{k\le K\\ k\equiv1\pmod m}} kp_k \right)^2 =: \alpha_0(K)$. Clearly, $\alpha_0(K)>\frac{p_1^2}{2\mu}$. Since $p_1=\mu e^{-\mu}$ for the Poisson distribution (which governs large sparse Erd\H{o}s--R\'enyi random networks) and $p_1=\frac{\mu}{(1+\mu)^2}$ for exponential distribution, we may take $\alpha=\frac{\mu}{2}e^{-2\mu}$ for Erd\H{o}s--R\'enyi random networks and $\alpha=\frac{\mu}{2(1+\mu)^4}$ for exponential networks in condition~(ii). $\alpha$ for networks with binomial and lognormal distributions can be obtained analogously.}. Numerical simulations on Erd\H{o}s--R\'enyi random networks (Figure~1A) and random networks with exponential degree distributions (Figure~1B) further support this theoretical prediction. In both cases, the probability of detecting a PNE using the asynchronous best-response procedure decreases with network size, which is consistent with the nonexistence result in Theorem 2\footnote{We use asynchronous best-response dynamics to detect PNE. \citet{Ramazi2016} prove that, in the absence of conformist-rebel edges, such dynamics converge to a PNE in finite time (the convergence rate is approximately $10n$ updates). For general CRS games, convergence is not guaranteed, but the dynamics can still serve as a practical detection heuristic. Consequently, the reported probability of PNE detection should be interpreted as a lower bound on the true probability of PNE existence: a PNE may exist even when the dynamics fail to converge within the allotted number of updates.}.

\begin{figure}[htbp]
    \centering
    \includegraphics[width=1\textwidth]{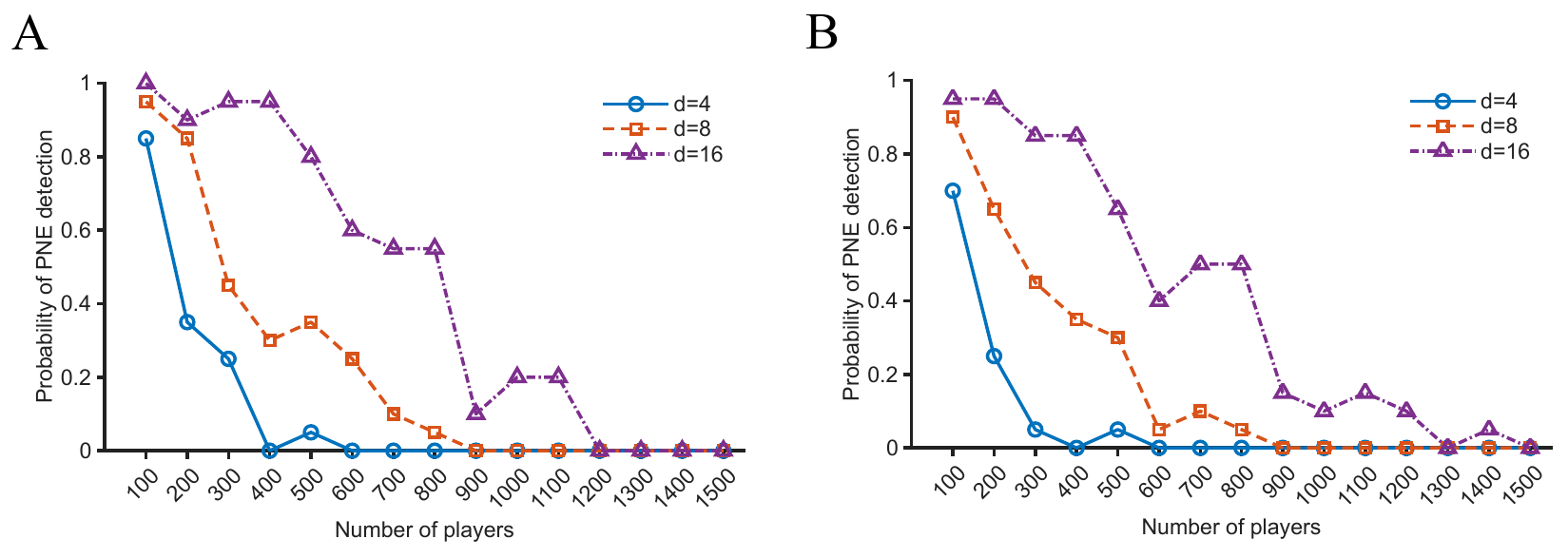}
    \caption{Probability of PNE detection as the number of players varies across different network degrees. (A) Erd\H{o}s--R\'enyi random networks. (B) Random networks with exponential degree distributions.  For each $(n,d)$, $20$ independent 3-strategy CRS game instances are generated, with $n=100,200,\ldots,1500$ and $d=4,8,$ or $16$. Player types are independently assigned with probabilities $\pi_C=\pi_R=0.25$ and $\pi_{S_k}=\frac{1}{6}$ for $k=1,2,3$. For each instance, asynchronous best-response dynamics are run from $10$ random initial profiles for at most $100n$ updates per run. The probability is the fraction of independently generated game instances for which a PNE is detected from at least one of the 10 random initial profiles.}
\end{figure}

\section{Equilibrium analysis: Special network structures}

In this section, we consider five special network structures, namely the complete network, line, ring, tree, and star. For these networks, we can provide sufficient and necessary conditions for the existence of a PNE and characterize the frequencies of strategies at a PNE.

\subsection{Complete network}

For a complete network, the payoff of a player is affected by the strategies of all other players. We denote the total number of players choosing strategy $a_k$ by $\chi_k:= |\{v \in N: x_v = a_k\}|$, and denote the total number of stubborn players choosing strategy $a_k$ by $s_k$. Without loss of generality, we assume $s_1 \leq s_2 \leq \cdots \leq s_m$, and let $\lceil x\rceil$ denote the ceiling function. The following lemma characterizes the strategies of conformists and rebels at a PNE.

\begin{lemma}\label{lem2}
At a PNE of an $m$-strategy CRS game on a complete network, \\
(i) all conformists choose the strictly most common strategy, i.e., they adopt strategy $a_i$ if and only if $\chi_i>\chi_j$ for all $j\neq i$, \\
(ii) if a rebel adopts strategy $a_i$, then $\chi_i\leq \chi_j+1$ for all $j\neq i$. In particular, if both $a_i$ and $a_j$ are adopted by rebels, then $|\chi_i-\chi_j|\leq 1$.
\end{lemma}

From Lemma \ref{lem2}, we derive necessary and sufficient conditions for the existence of a PNE in a $m$-strategy CRS game on a complete network.


\begin{theorem}\label{the3}
Suppose that $s_1\leq s_2\leq\cdots\leq s_m$ and $c\geq1$. The $m$-strategy CRS game on complete network admits a PNE if and only if one of the following conditions holds:\\ 
(i) $s_m+c>\left\lceil\left(n-c-s_m\right) /(m-1)\right\rceil$, \\
(ii) $s_m+c \leq\left\lceil\left(n-c-s_m\right) /(m-1)\right\rceil$ and $n \equiv 1 \pmod{m}$.
\end{theorem}

\begin{proof}
\emph{Sufficiency.} Case (i): Following the idea of partial potential
analysis \citep{Zhang2018}, set $x_i=a_m$ for every $i\in C$ and distribute
the rebels sequentially among strategies $a_1,\ldots,a_{m-1}$, each time
assigning a rebel to a strategy with the currently smallest total frequency.
Since $s_i\leq s_m$ for all $i<m$, the resulting allocation satisfies $\max_{i<m}\chi_i \leq \max\left\{s_m, \left\lceil \frac{n-c-s_m}{m-1} \right\rceil \right\} <s_m+c=\chi_m$. Thus, $a_m$ is uniquely most common, and all conformists are best responding by Lemma~\ref{lem2}(i). The balancing construction also ensures that every strategy chosen by a rebel has frequency at most one larger than any other strategy except $a_m$, so Lemma~\ref{lem2}(ii) implies that all rebels are best responding. Hence, the constructed profile is a PNE.

\textit{Case (ii).}
Since $n \equiv 1 \pmod{m}$, write $n=km+1$ for some integer $k \geq 1$.
The condition $s_m+c\leq \left\lceil\left(n-c-s_m\right) /(m-1)\right\rceil$ implies $s_m+c\leq k$. Indeed, if $s_m+c\geq k+1$,
then $\left\lceil\frac{km+1-s_m-c}{m-1}\right\rceil\leq k<s_m+c$, a contradiction. Again assign strategy $a_m$ to all conformists. Let $k+1-s_m-c$ rebels choose $a_m$, and, for each $i<m$, let $k-s_i$ rebels choose $a_i$. These numbers are nonnegative because $s_i\leq s_m\leq s_m+c\leq k$, and their sum equals $r$. The resulting strategy frequencies satisfy $\chi_m=k+1$, $\chi_i=k$ for all $i<m$. Thus, by Lemma~\ref{lem2}(i), the conformists choosing $a_m$ are playing best responses. For every rebel, Lemma~\ref{lem2}(ii) is also satisfied, since the frequencies of any two strategies differ by at most one. Hence the constructed profile is a PNE.

\emph{Necessity.} Suppose a PNE exists in a case other than (i) or (ii).
By Lemma \ref{lem2}(i), all conformists choose the same strategy $a_v$,
and the number of players adopting $a_v$ strictly exceeds that of any other
strategy. The condition $s_m+c\leq \lceil(n-c-s_m)/(m-1)\rceil$ implies that even if all conformists pick $a_v$, they cannot make it the strictly most common strategy without any rebel choosing $a_v$. Thus, at least one rebel must also select $a_v$. For such a rebel, Lemma~\ref{lem2}(ii) implies that
$\chi_v\leq\chi_j+1$ for all $j\neq v$. Since $a_v$ is strictly the most
common strategy, $\chi_v>\chi_j$ for all $j\neq v$. Hence,
$\chi_j=\chi_v-1$ for all $j\neq v$. Therefore, $n=\chi_v+(m-1)(\chi_v-1)=m(\chi_v-1)+1$, which implies $n\equiv1\pmod m$. This contradicts the failure of condition~(ii). Thus, a PNE can exist only under cases~(i) and~(ii).
\end{proof}

Theorem \ref{the3} provides necessary and sufficient conditions for PNE existence on complete networks, but PNE are generally non-unique. Multiplicity arises for two reasons. First, under a given equilibrium strategy frequency vector $\mathbf{\chi}=(\chi_1,...,\chi_m)$, rebels can be distributed among strategies in different ways without altering that vector, resulting in distinct equilibrium profiles. Second, the equilibrium frequency vector $\mathbf{\chi}$ itself need not be unique. When multiple strategies attain the same maximal frequency, conformists can coordinate on any of them, producing alternative equilibrium distributions. Consequently, the number of PNE typically grows with population size, and the set of possible equilibrium frequency vectors depends on the number of strategies and numbers of the three player types. Specifically, for the $2$-strategy case, the strategy frequency at a PNE can be completely characterized by the proof of Theorem \ref{the3}.

\begin{corollary}
The $2$-strategy CRS game on complete network admits a PNE if and only if one of the following conditions holds: (i) $\left|s_2-s_1\right|+c>r$, (ii) $\left|s_2-s_1\right|+c \leq r$ and $n$ is odd. If condition (i) holds, then the equilibrium strategy frequency is $\mathbf{\chi}=(r+s_1, c+s_2)$ if $s_2\geq s_1$ and is $\mathbf{\chi}=(c+s_1, r+s_2)$ if $s_2\leq s_1$. If condition (ii) holds, then the equilibrium strategy frequency is $\mathbf{\chi}=(\dfrac{n+1}{2}, \dfrac{n-1}{2})$ or $\mathbf{\chi}=(\dfrac{n-1}{2}, \dfrac{n+1}{2})$.
\end{corollary}

\subsection{Line and ring}

In this section, we examine two closely related network structures, lines and rings. A game $G$ is played on a line if players can be labeled so that $E=\{12, 23, \dots, (n-1)n\}$, and on a ring if $E = \{12, 23, \dots, (n-1)n, n1\}$. We require $n \geq 3$ so that a ring is feasible. In addition, we focus on $m\geq 3$ since the case of 2-strategy has been analyzed in \citet{Cao2019}. The following theorem gives necessary and sufficient conditions for a PNE on these networks.

\begin{theorem}\label{the4}
An $m$-strategy CRS game ($m\geq 3$) on a line or a ring admits a PNE if and only if every conformist has at least one neighbor who is either a conformist or a stubborn agent (i.e., $N_i \cap (C \cup S) \neq \emptyset$ for all $i \in C$).
\end{theorem}

\begin{proof}
\medskip
\emph{Sufficiency.} Suppose every conformist has a neighbor in $C \cup S$. Consider each maximal connected component of conformists. If the component is adjacent to a stubborn agent, assign all conformists in it the strategy of that stubborn agent (choosing arbitrarily if there are several). Otherwise, the component has no stubborn neighbor and therefore, by assumption, contains at least two conformists; assign them an arbitrary common strategy. In either case, each conformist shares a strategy with at least one neighbor and has no incentive to deviate. Treating these conformists as stubborn reduces the game to an RS game, which is a potential game and hence admits a PNE. Thus, the original CRS game has a PNE.

\medskip
\emph{Necessity.} Suppose that a PNE $\mathbf{x}^*$ exists but a conformist $i$ has only rebel neighbors. If $i$ is an endpoint of a line, let $j$ be its unique neighbor. If $x_i^* \neq x_j^*$, then $i$ would profit by switching to $x_j^*$. If $x_i^* = x_j^*$, rebel $j$ can profitably deviate to a strategy different from both of its neighbors. Hence, no such PNE can exist. If $i$ is an interior node with rebel neighbors $j$ and $k$, then in equilibrium we must have $x_i^* \neq x_j^*$ and $x_i^* \neq x_k^*$. Otherwise, the rebel whose strategy matches $x_i^*$ could profitably deviate. But then conformist $i$ receives zero payoff and could improve by adopting either $x_j^*$ or $x_k^*$, a contradiction. Therefore, any conformist must have at least one neighbor in $C \cup S$.
\end{proof}

Figure 2 illustrates PNE in line and ring networks for $3$-strategy CRS game. By Theorem \ref{the4}, a CRS game on a ring does not have a PNE if and only if some conformists have two rebel neighbors (an $R$-$C$-$R$ pattern). On a line, a PNE fails to exist if and only if some conformists either occupy an endpoint whose only neighbor is a rebel or has two rebel neighbors. 

\begin{figure}[htbp]
    \centering
    \includegraphics[width=0.6 \textwidth]{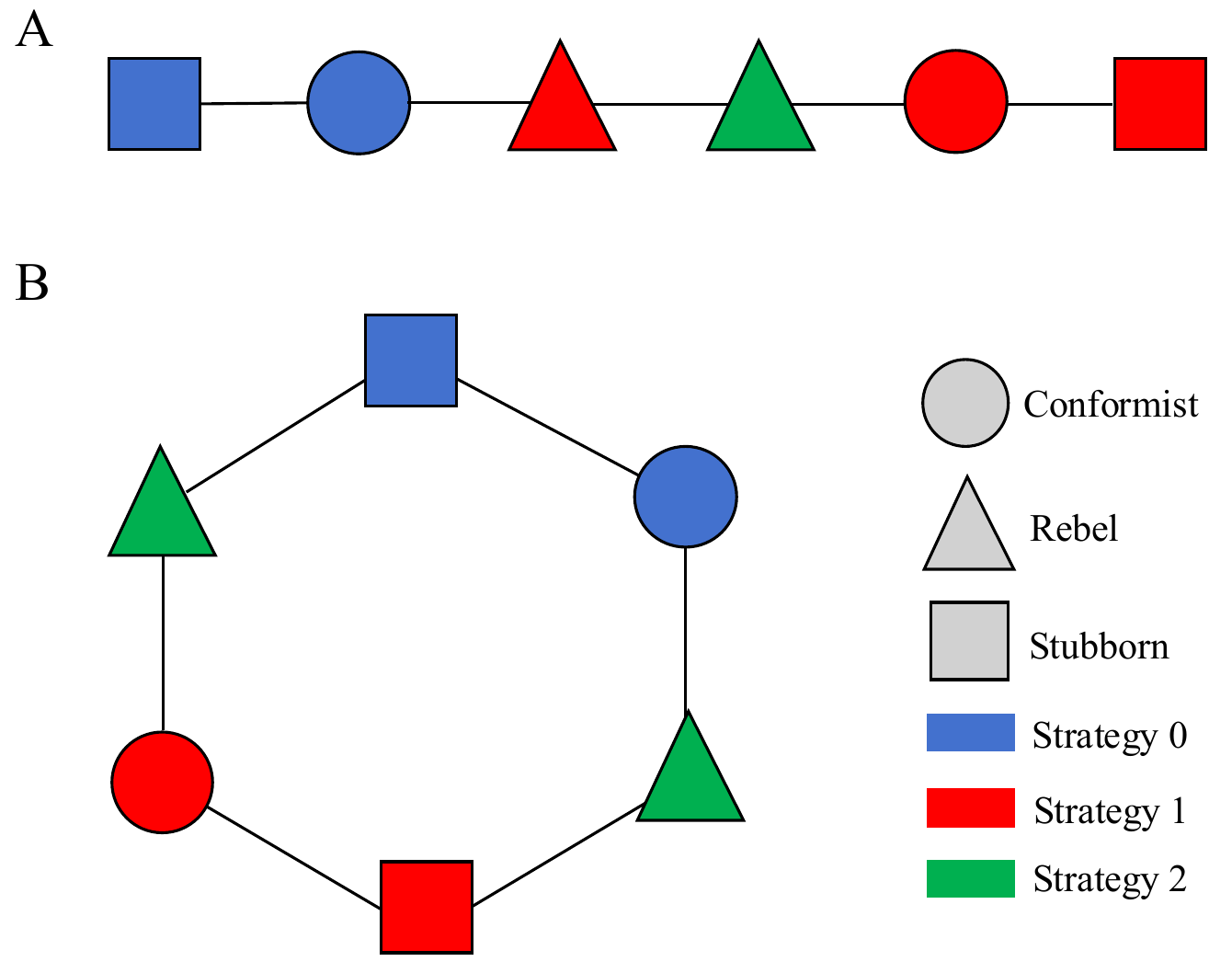}
    \caption{Examples of PNE in 3-strategy networked CRS games. (A) A PNE on a line. (B) A PNE on a ring. Circles, triangles, and squares represent conformist, rebel, and stubborn agents, respectively; blue, red, and green correspond to strategies 0, 1, and 2, respectively.}
    \label{fig:line-graph}
\end{figure}


\subsection{Tree}

Trees, characterized by the absence of cycles, are another widely studied class of networks. The following theorem provides a necessary condition for a PNE when the number of strategies exceeds the maximum degree of the tree (i.e., $m> \Delta$).

\begin{theorem}\label{the5}
If an $m$-strategy CRS game on a finite tree with maximum degree $\Delta<m$ admits a PNE, then every conformist has a neighbor in $C\cup S$.
\end{theorem}

\begin{proof}
\medskip

Suppose a PNE $\mathbf{x}^*$ exists but a conformist $i$ has only rebel neighbors. For any $j \in N_i$, rebel $j$ can adopt a strategy different from all the strategies of its neighbors since $m > \Delta$. Hence, in a PNE, we must have $x_i^* \neq x_j^*$ for every $j\in N_i$, which yields conformist $i$ a payoff of zero. Yet, conformist $i$ could switch to any neighbor's strategy and obtain a payoff of at least one, a contradiction. Therefore, every conformist must have a neighbor in $C\cup S$.
\end{proof}


We note that the condition in Theorem~\ref{the5} is necessary but not sufficient. A CRS game on a tree may fail to admit a PNE even if every conformist has a neighbor in $C\cup S$. Figure~3B illustrates this with a simple example where $m=4$ and $\Delta=3$. In this example, the rebel and the conformist cannot simultaneously satisfy their best-response conditions, and hence no PNE exists. The degree condition $\Delta<m$ is also essential for the necessity result. Figure~3C provides a counterexample with $m=4$ and $\Delta=5>m$, where a PNE exists despite a conformist having only a rebel neighbor. In the displayed PNE, the conformist and the rebel choose the same strategy, while the remaining neighbors of the rebel are stubborn agents whose fixed strategies make the conformist's strategy a least frequent strategy in the rebel's neighborhood. Thus, once $\Delta\geq m$, the condition that every conformist has a neighbor in $C\cup S$ is no longer necessary.

\begin{figure}[htbp]
    \centering
    \includegraphics[width=1 \textwidth]{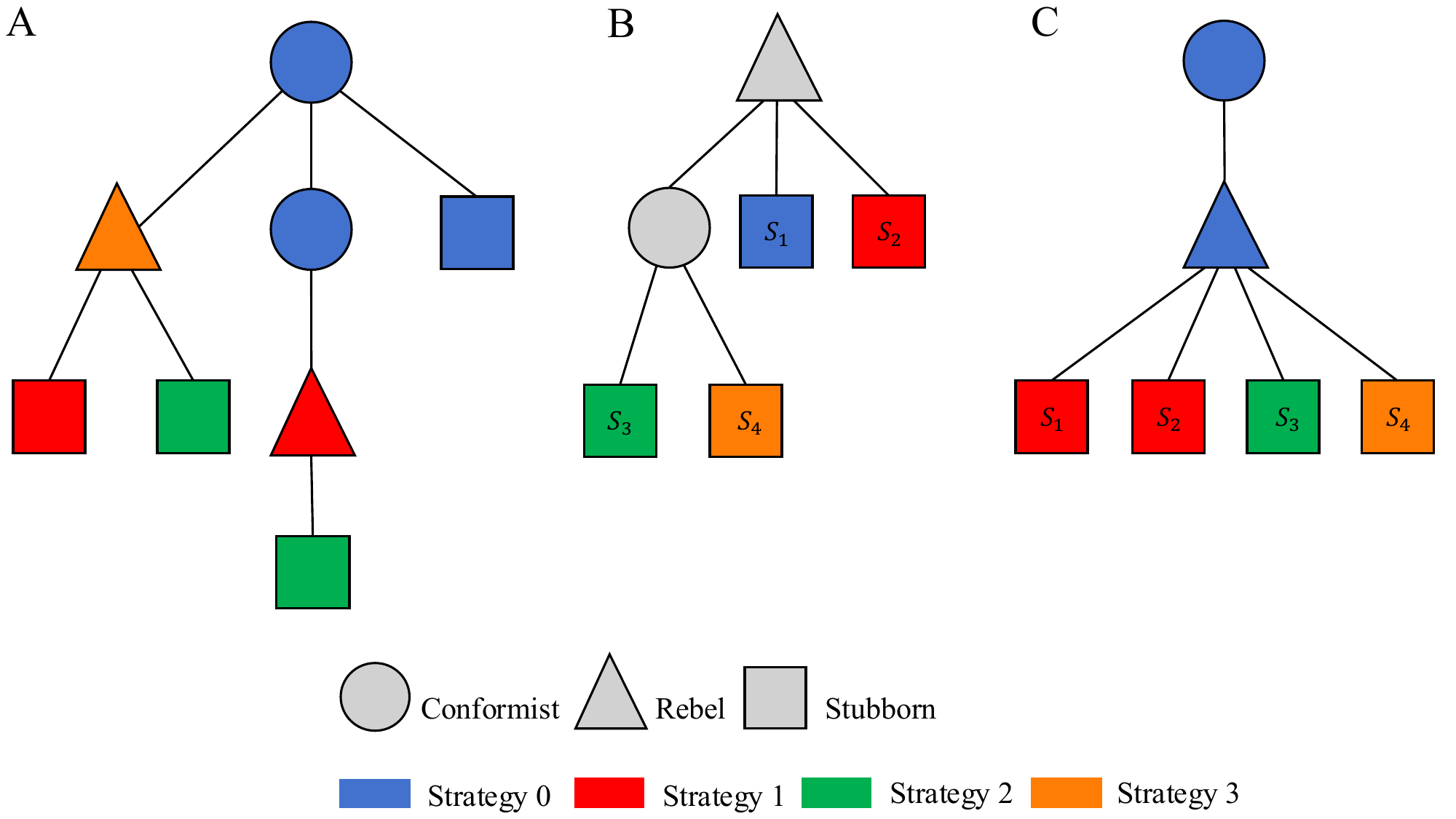}
    \caption{Examples of CRS games on tree networks. (A) A tree with $\Delta<m$ admitting a PNE. (B) A tree with $\Delta<m$ admitting no PNE, although every conformist has a neighbor in $C\cup S$. (C) A tree with $\Delta>m$ admitting a PNE, although the conformist has only a rebel neighbor. Circles, triangles, and squares represent conformist, rebel, and stubborn agents, respectively. The colors indicate the assigned strategies.} 
    \label{fig:tree-graph}
\end{figure}

We next provide an algorithm that determines whether a PNE exists on a finite tree and constructs one if it does (full definitions and all details are provided in the \ref{tree-algorithm}). The algorithm roots the tree at an arbitrary node and processes nodes from leaves to root. For each non-root node, conditional on its parent's strategy, we compute the set of strategies that allow the node's entire subtree to be completed with all players playing best responses. At a leaf, this set is determined by its type. At an internal node, we check all feasible strategy combinations of its children; a candidate strategy is accepted if there exists a combination under which the node's own strategy is a best response. Specifically, conformists require their strategy to be among the most frequent, whereas rebels require theirs to be among the least frequent. Stubborn nodes impose no best-response constraint of their own but require all child subtrees to be feasible. After processing all non-root nodes, the root is checked similarly, except that it has no parent. The game admits a PNE if and only if a feasible root strategy exists. The algorithm then reconstructs the full profile from the stored witnesses; if no such strategy exists, it reports no PNE.

Regarding algorithm complexity, enumeration at each node is bounded by the maximum degree $\Delta$ and strategy count $m$. Since each node has at most $\Delta-1$ children, the number of child combinations per node is $m^{\Delta-1}$, each checked in $O(\Delta+m)$. The total time is $O(n(\Delta+m)m^{\Delta+1})$, which is linear in $n$ when $m$ and $\Delta$ are fixed. 

\subsection{Star}

A special case of tree is star. A game $G$ is played on a star if there exists a player $i$ such that $E=\{ij: j\in N, j\neq i\}$, where $i$ is the central player and the other $n-1$ players are leaves. Theorem \ref{the5} implies that when $m>n-1$, a CRS game on a star admits a PNE if and only if one of the following holds: (i) the central player is a conformist and not all leaves are rebels; (ii) the central player is a rebel and no leaf is a conformist; (iii) the central player is stubborn.

The star admits a sharper PNE characterization when $m \leq n-1$. In this case, the existence of a PNE depends not only on the type of central player, but also on the composition of leaf types. Let $G$ be a star network with central player $i$ and leaf set $L=N_i(G)$. For each strategy $a\in A$, let $s_a$ denote the number of stubborn leaves fixed at $a$, and let $c_L$ and $r_L$ denote the numbers of conformist and rebel leaves, respectively.

\begin{figure}[htbp]
    \centering
    \includegraphics[width=0.8\textwidth]{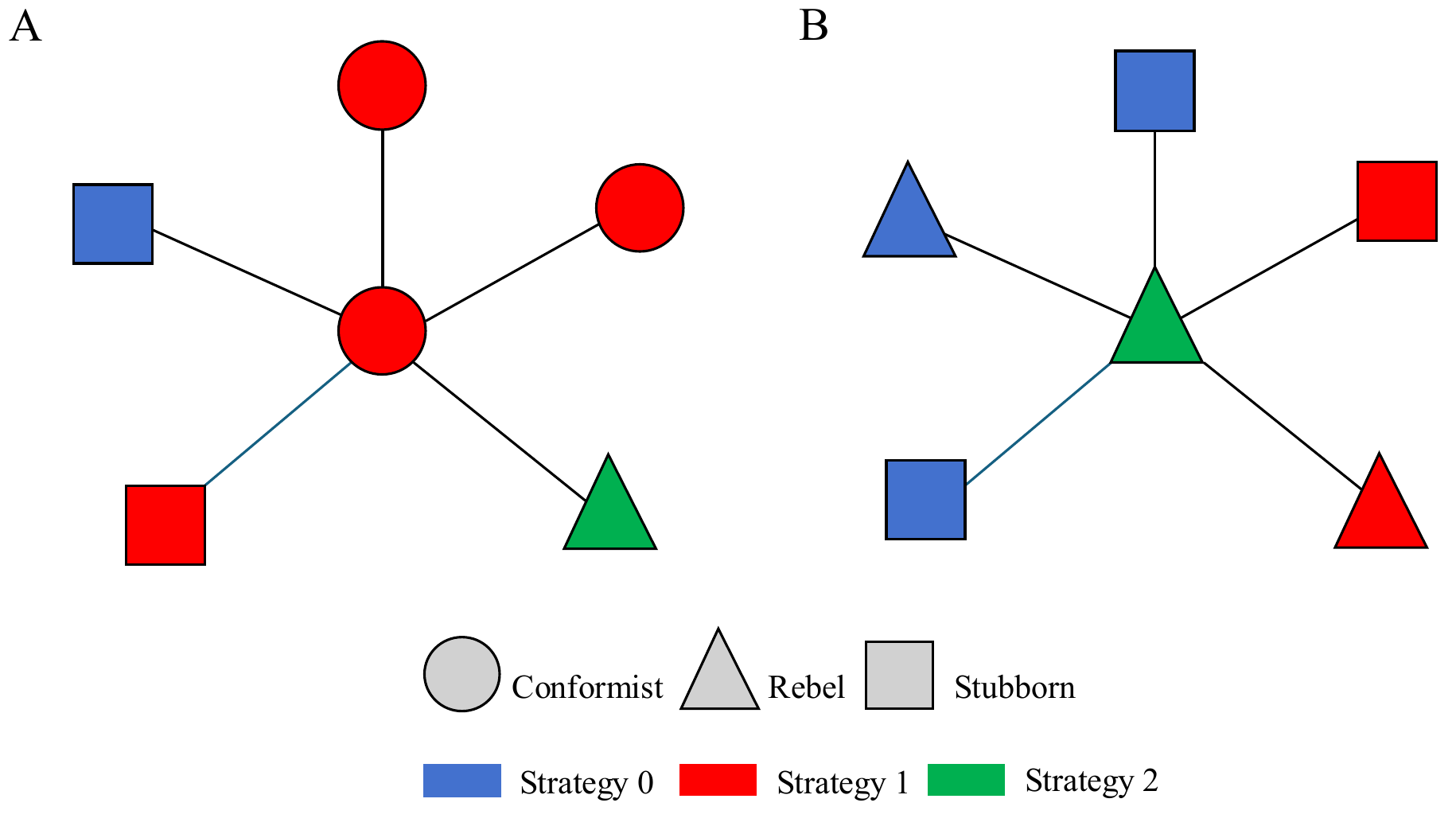}
    \caption{Examples of PNE in 3-strategy CRS games on star networks. (A) A PNE with a conformist at the center. (B) A PNE with a rebel at the center. Circles, triangles, and squares represent conformist, rebel, and stubborn agents, respectively; blue, red, and green correspond to strategies 0, 1, and 2, respectively.}
    \label{fig:star-graph}
\end{figure}

\begin{corollary}\label{cor2}
An $m$-strategy CRS game on a star admits a PNE if and only if one of the following conditions holds:
(i) the central player is a conformist and there exists $a\in A$ such that $r_L\leq \sum_{b\neq a} \max\{0,c_L+s_a-s_b\}$;
(ii) the central player is a rebel and there exists $a\in A$ such that $r_L\geq \sum_{b\neq a} \max\{0,c_L+s_a-s_b\}$;
(iii) The central player is stubborn.
\end{corollary}

\begin{proof}
(i) Suppose that the central player is a conformist and plays $a\in A$. At any PNE, every conformist leaf also plays $a$, while every rebel leaf chooses a strategy different from $a$. Thus, strategy $a$ appears $c_L+s_a$ times among leaves. For any $b\neq a$, at most $\max\{0,c_L+s_a-s_b\}$ rebel leaves can be assigned strategy $b$ without making $b$ more frequent than $a$. Hence, all rebel leaves can be accommodated in this way if and only if
$r_L\le \sum_{b\in A\setminus{a}} \max\{0,c_L+s_a-s_b\}$. Therefore, a PNE exists if and only if this inequality holds for some $a\in A$ (Figure 3A).

(ii) Suppose the central player is a rebel and plays  $a\in A$. An analogous counting argument shows that strategy $a$ is not the most frequent leaf strategy if and only if
$r_L\ge \sum_{b\in A\setminus{a}} \max\{0,c_L+s_a-s_b\}$. Consequently, a PNE exists if and only if this inequality is satisfied for some $a\in A$  (Figure 3B).

(iii) If the central player is stubborn, fix its strategy. Assign every conformist leaf that same strategy and every rebel leaf a different one. The resulting profile is a PNE.
\end{proof}

\section{Discussion}

\subsection{Concluding remarks}
Our paper advances the understanding of strategic interactions in $m$-strategy network games with the coexistence of strategic complements and strategic substitutes by making two main contributions. First, we derive sufficient conditions for the existence of a PNE on arbitrary networks, while proving that equilibria almost surely vanish in large random networks. Second, we derive necessary and sufficient conditions for PNE existence and fully characterize the equilibrium strategy frequencies for complete network, lines, rings, trees, and stars. Collectively, these results offer a unified perspective that PNE is likely to exist when every conformist has more conformist and stubborn neighbors, and fails when the network game has numerous CR edges. This perspective is consistent with previous studies on 2‑strategy CR and CRS games \citep{Zhang2018,Cao2026}

\subsection{Applications}

A direct application of our model is opinion formation and multiparty elections on social networks \citep{Acemoglu2013,Yildiz2013,Stewart2019,Cao2024}. In this setting, the strategies represent different opinions, candidates, parties, or policy positions. Conformists are individuals who tend to follow the locally popular choice, rebels are individuals who prefer to distinguish themselves from their neighbors, and stubborn agents are committed supporters or opinion leaders whose choices are fixed. A PNE then represents a locally stable opinion or voting configuration. Our results show that such stability depends not only on the numbers of different types of agents, but also on their local arrangement in the network. For example, a conformist surrounded only by rebels may become a local source of instability, while a stubborn neighbor can serve as a local anchor and help restore equilibrium.

Another application is fashion choice or market competition with multiple products \citep{Jackson2008,Zhang2018}. The strategies may represent different clothing styles, brands, platforms, or technologies. Some individuals imitate the locally popular choice, some prefer to be different, and others remain loyal to a fixed choice. Under this interpretation, our CRS model describes the coexistence of imitation, differentiation, and commitment in the same social network. The equilibrium results characterize when stable multi-choice patterns can emerge and when local conformist-rebel interactions may lead to fashion cycles.

\subsection{Possible extensions}
Although our work assumes that conformists and rebels treat all strategies symmetrically, several findings extend to the more general case in which they have strategy specific preferences. Let $v_a>0$ and $w_a>0$ be
strategy-specific weights associated with strategy $a\in A$. In this variant, the payoffs for the three player types can be written as 
\[
u_i(\mathbf{x})=
\begin{cases}
v_{x_i}\left|\{j\in N_i:x_j=x_i\}\right|, & i\in C,\\[1mm]
w_{x_i}\left|\{j\in N_i:x_j\neq x_i\}\right|, & i\in R,\\[1mm]
0, & i\in S.
\end{cases}
\]
A larger $v_a$ makes strategy $a$ more attractive for conformists, while
a larger $w_a$ makes strategy $a$ more attractive for rebels when they
disagree with their neighbors. Theorem \ref{the1}(i) (together with Result 1 and Result 2) continues to hold because a PNE is guaranteed in games without CR edges, with potential functions for the CS and RS games given by 
$$
\phi_{CS}(\boldsymbol{x}):=\frac{1}{2}\left(\sum_{k \in C}v_{x_k}\left|\left\{j \in N_k: x_j = x_k\right\}\right|+\sum_{k \in C}v_{x_k}\left|\left\{j \in N_k \cap S: x_j = x_k\right\}\right|\right)
$$
and 
\[
\begin{aligned}
\phi_{RS}(\mathbf{x}):=\frac{1}{2}\Bigg(&\sum_{k\in R} w_{x_k}\left|\{j\in N_k:x_j\neq x_k\}\right| +\sum_{k\in R} w_{x_k}\left|\{j\in N_k\cap S:x_j\neq x_k\}\right|  \\ &+\sum_{k\in R} w_{x_k}|N_k\cap R| \Bigg),
\end{aligned}
\]
respectively. For any unilateral deviation by a conformist in the CS game or by a rebel in the RS game, the change in the corresponding potential equals the change in that player's payoff. Hence, both weighted CS and weighted RS
games are exact potential games and admit a PNE. In addition, Theorem \ref{the1}(ii) can also extend to this variant case. Let $v_{\max}:=\max_{b\in A}v_b.$ If there exists a strategy $a\in A$ such that each conformist has at least $\frac{v_{\max}}{v_a+v_{\max}}$ of his/her neighbors in $C \cup S_a$, then all conformists have no incentive to deviate, reducing the game to an RS game that admits a PNE. However, generalizing the remaining theorems to this variant requires further analysis.

In addition, Theorem~\ref{the1}(i) can be extended to undirected weighted networks. Suppose that each edge $\{i,j\}\in E$ has a symmetric positive weight $\omega_{ij}=\omega_{ji}>0$. For any $A\subseteq N_i$, define its weighted cardinality with respect to player $i$ by $|A|_{\omega_i}:=\sum_{j\in A}\omega_{ij}$. In the weighted CRS game, conformists and rebels choose weighted best responses to their neighbors' strategies. The weighted utility function can be written as
\[
u_i(\mathbf{x})=
\begin{cases}
\left|\{j\in N_i:x_j=x_i\}\right|_{\omega_i}, & i\in C,\\[1mm]
\left|\{j\in N_i:x_j\neq x_i\}\right|_{\omega_i}, & i\in R,\\[1mm]
0, & i\in S.
\end{cases}
\]
When $\omega_{ij}=1$ for every $\{i,j\}\in E$, the weighted CRS game reduces to the original unweighted CRS game. When there are no conformist-rebel edges, conformists interact only with conformists and stubborn agents, while rebels interact only with rebels and stubborn agents. Hence, the game can be decomposed into a weighted CS game and a weighted RS game. For these two games, potential functions are given by
\[\phi_{CS}(\mathbf{x}) = \frac{1}{2} \left[ \sum_{i\in C} \left|\{j\in N_i:x_j=x_i\}\right|_{\omega_i} + \sum_{i\in C} \left|\{j\in N_i\cap S:x_j=x_i\}\right|_{\omega_i} \right]\]
and
\[\phi_{RS}(\mathbf{x})= \frac{1}{2}\left[\sum_{i\in R}\left|\{j\in N_i:x_j\neq x_i\}\right|_{\omega_i}+\sum_{i\in R}\left|\{j\in N_i\cap S:x_j\neq x_i\}\right|_{\omega_i}\right],\]
respectively. Thus, both the weighted CS game and the weighted RS game are exact potential games. Therefore, if the weighted CRS game has no conformist-rebel edges, then $\phi(\mathbf{x}) := \phi_{CS}(\mathbf{x}) + \phi_{RS}(\mathbf{x})$ is an exact potential function for the whole game. Since the strategy space is finite, $\phi$ has at least one maximizer, and every maximizer is a PNE. Hence Theorem~\ref{the1}(i) also holds for undirected weighted networks.

\subsection{Future studies}

First, while our paper establishes necessary and sufficient conditions for the existence of PNE on specific network topologies, the question of uniqueness remains open. Prior work on strategic complements and substitutes links equilibrium uniqueness to the lowest eigenvalue of the network matrix \citep{Bramoulle2014}. A natural direction for future research is therefore to characterize PNE uniqueness conditions in CRS games, possibly building on such spectral properties.

Second, the present analysis is static. A natural extension is to study dynamic adjustment processes, such as best-response dynamics, asynchronous updating, imitation dynamics, or other evolutionary dynamics on networks \citep{Acemoglu2013,Ramazi2016,Zhang2018,Como2020,Le2023,Ramazi2023, Pei2024,Cao2026,Wang2026}. Such models would help determine whether PNE are reachable, whether the dynamics converge, and when conformist-rebel interactions generate cycles, persistent fluctuations, or polarization.

Third, future work could incorporate additional behavioral types. Beyond conformists, rebels, and stubborn agents, real networks may include best responders, imitators, and Bayesian players \citep{Chandrasekhar2020, Le2020,Le2023,Ding2026}. Studying such mixed populations would provide a more general framework for heterogeneous strategic interaction on networks.

Finally, another promising direction is to extend CRS games from finite strategy sets to continuous strategy spaces. Continuous-action network games are widely used to model opinion formation, peer effects, local public goods, effort, investment, and technology adoption \citep{Bramoulle2014,Belhaj2014}. Our previous work on the continuous-strategy CR model shows that the presence of a few rebels drives group opinions toward the golden mean, a result that follows from standard matrix theory \citep{Cao2018}. Extending this analysis to CRS games with all three player types would require new tools to establish equilibrium existence, uniqueness, and stability.

\section*{Acknowledgments}
Wenjie Cao acknowledges support from the National Natural Science Foundation of China (No.724B2006) and the China Scholarship Council (No.\allowbreak 202406040147). Boyu Zhang acknowledges support from the National Natural Science Foundation of China (No.72131003, No.72573024, and No.72671038) and the Beijing Natural Science Foundation (No.Z220001).

\newpage

\newpage
\appendix

\section{PNE-testing algorithm on trees}\label{tree-algorithm}

\subsection{Algorithm description} \label{app:algorithm-description}

Let $G=(N,E)$ be a tree with $|N|=n$, and let $A=\{1,\ldots,m\}$ be the strategy set. For each stubborn agent $i\in S$, let $s_i\in A$ denote its fixed strategy. Choose an arbitrary root $\rho$ and orient the tree away from $\rho$. For each vertex $i$, let $\operatorname{ch}(i)$ denote the set of children of $i$. The algorithm proceeds from leaves to the root.

For each non-root vertex $i$ and each possible strategy $\alpha\in A$ of its
parent, we compute a set $D_i(\alpha)\subseteq A.$ The set $D_i(\alpha)$ contains exactly those strategies $\beta$ of $i$ such
that, when the parent of $i$ chooses $\alpha$ and $i$ chooses $\beta$, 
the subtree rooted at $i$ can be completed so that every player in this
subtree is playing a best response, taking the parent strategy $\alpha$ as a
boundary condition.

\medskip
\noindent\textit{Leaf vertices.}
If $i$ is a leaf and its parent chooses strategy $\alpha$, then
$$
D_i(\alpha)=
\begin{cases}
\{s_i\}, & i\in S,\\[1mm]
\{\alpha\}, & i\in C,\\[1mm]
A\setminus\{\alpha\}, & i\in R.
\end{cases}
$$
Indeed, a stubborn leaf keeps its fixed strategy, a conformist leaf must match its parent, and a rebel leaf must choose a strategy different from its parent. The set $D_i(\alpha)$ is conditional. It does not mean that the strategy of $i$ has already been fixed. Rather, it records all strategies of $i$ that are feasible under the hypothetical condition that the parent of $i$ chooses $\alpha$. The best-response condition of the parent of $i$ is not checked at this step; it is checked later when the parent vertex is processed. Thus, a strategy of $i$ that is feasible for its own subtree may still be rejected at the parent level if it prevents the parent from being a best response.

\medskip
\noindent\textit{Internal non-root vertices.}
Let $i$ be an internal non-root vertex. Suppose that all children of $i$ have
already been processed. Fix the parent strategy $\alpha\in A$ and a candidate
strategy $\beta\in A$ for node $i$. For each child $j\in\operatorname{ch}(i)$, the feasible strategies of $j$ are given by $D_j(\beta).$ Thus, if $i$ chooses $\beta$, then child $j$ may choose any strategy $\gamma_j\in D_j(\beta)$. For a choice of child strategies $\gamma_j\in D_j(\beta), j\in\operatorname{ch}(i),$ define $q_a= \mathbf{1}_{\{\alpha=a\}} +
\left|\{j\in\operatorname{ch}(i):\gamma_j=a\}\right|,$ where $a\in A$. Here $q_a$ is the number of neighbors of $i$ choosing strategy $a$.

\textit{Case 1: $i\in S$.}
If $i$ is stubborn, then it must choose its fixed strategy $s_i$.
Therefore,
$$
D_i(\alpha)=
\begin{cases}
\{s_i\}, & \text{if } D_j(s_i)\neq \emptyset
\text{ for every } j\in\operatorname{ch}(i),\\[1mm]
\emptyset, & \text{otherwise}.
\end{cases}
$$
That is, a stubborn node itself imposes no best-response constraint, but its
children must be feasible when the stubborn node plays $s_i$.

\textit{Case 2: $i\in C$.}
If $i$ is a conformist, then a candidate strategy $\beta$ is feasible, i.e.,
$\beta\in D_i(\alpha)$, if and only if there exists a choice $\gamma_j\in D_j(\beta), j\in\operatorname{ch}(i),$ such that $q_\beta=\max_{a\in A} q_a.$

\textit{Case 3: $i\in R$.}
If $i$ is a rebel, then a candidate strategy $\beta$ is feasible, i.e.,
$\beta\in D_i(\alpha)$, if and only if there exists a choice $\gamma_j\in D_j(\beta), j\in\operatorname{ch}(i),$ such that $q_\beta=\min_{a\in A} q_a.$

\medskip
\noindent\textit{Root vertex.}
After all non-root vertices have been processed, we check the root $\rho$.
Since the root has no parent, for a candidate root strategy $\beta$ and a
choice of child strategies $\gamma_j\in D_j(\beta), j\in\operatorname{ch}(\rho),$ define $q_a= \left|\{j\in\operatorname{ch}(\rho):\gamma_j=a\}\right|$, $a\in A.$ If $\rho\in S$, then $\beta$ must be $s_\rho$, and all child subtrees must be feasible. Hence a PNE exists if and only if $D_j(s_\rho)\neq\emptyset$ for every $j\in\operatorname{ch}(\rho)$. If $\rho\in C$, then a PNE exists if and only if there exist $\beta\in A$ and $\gamma_j\in D_j(\beta)$ for all $j\in\operatorname{ch}(\rho)$ such that $q_\beta=\max_{a\in A} q_a.$ If $\rho\in R$, then a PNE exists if and only if there exist $\beta\in A$ and $\gamma_j\in D_j(\beta)$ for all
$j\in\operatorname{ch}(\rho)$ such that $q_\beta=\min_{a\in A} q_a.$

The game admits a PNE if and only if there exists at least one feasible root strategy.

\subsection{Dynamic programming formulation}
\label{app:algorithm-flowchart}

\begin{algorithm}[H]
\scriptsize
\caption{PNE-testing algorithm on trees}
\label{alg:tree-PNE}

\KwIn{
An $m$-strategy CRS game on a finite tree with maximum degree $\Delta$.
}
\KwOut{
A PNE $\mathbf{x}$ if one exists; otherwise, No PNE.
}

\BlankLine
\textbf{Bottom-up computation of feasible sets}\;

\ForEach{non-root vertex $i$ in postorder}{
    \ForEach{parent strategy $\alpha\in A$}{
        Initialize $D_i(\alpha)\leftarrow \emptyset$\;

        \uIf{$i$ is a leaf}{
            \uIf{$i\in S$}{
                $D_i(\alpha)\leftarrow \{s_i\}$\;
            }
            \uElseIf{$i\in C$}{
                $D_i(\alpha)\leftarrow \{\alpha\}$\;
            }
            \ElseIf{$i\in R$}{
                $D_i(\alpha)\leftarrow A\setminus\{\alpha\}$\;
            }
        }

        \Else{
            \uIf{$i\in S$}{
                \If{$D_j(s_i)\neq \emptyset$ for every
                $j\in\operatorname{ch}(i)$}{
                    $D_i(\alpha)\leftarrow \{s_i\}$\;
                    Store one $\gamma_j\in D_j(s_i)$ for each
                    $j\in\operatorname{ch}(i)$\;
                }
            }

            \Else{
                \ForEach{candidate strategy $\beta\in A$}{
                    \ForEach{choice of child strategies
                    $\gamma_j\in D_j(\beta)$ for all
                    $j\in\operatorname{ch}(i)$}{
                        Compute
                        \[
                        q_a=
                        \mathbf{1}_{\{\alpha=a\}}
                        +
                        \left|\{j\in\operatorname{ch}(i):\gamma_j=a\}\right|,
                        \qquad a\in A .
                        \]

                        \uIf{$i\in C$ and $q_\beta=\max_{a\in A}q_a$}{
                            Add $\beta$ to $D_i(\alpha)$\;
                            Store this child-strategy choice as a witness\;
                            \textbf{break}\;
                        }

                        \uElseIf{$i\in R$ and $q_\beta=\min_{a\in A}q_a$}{
                            Add $\beta$ to $D_i(\alpha)$\;
                            Store this child-strategy choice as a witness\;
                            \textbf{break}\;
                        }
                    }
                }
            }
        }
    }
}
\BlankLine
\textbf{Root check}\;

Process the root $\rho$ in the same way as an internal vertex, except that it
has no parent strategy\;

\uIf{there is no feasible root strategy}{
    \Return{No PNE}\;
}
\Else{
    Choose one feasible root strategy\;
    Recursively assign strategies to the children using the stored witnesses\;
    \Return{$\mathbf{x}$}\;
}

\setcounter{AlgoLineStore}{\value{AlgoLine}}
\end{algorithm}

\subsection{Proof of correctness} \label{app:algorithm-proof}

\begin{proposition}
The above algorithm determines whether the $m$-strategy CRS game on
a finite tree admits a PNE.
\end{proposition}

\begin{proof}
Choose an arbitrary root $\rho$ and orient the tree away from $\rho$. For each non-root vertex $i$, let $T_i$ denote the subtree rooted at $i$. We claim that the algorithm maintains the following invariant: for every parent strategy  $\alpha\in A$, $\beta\in D_i(\alpha)$ if and only if, conditional on the parent of $i$ choosing strategy $\alpha$ and $i$ choosing strategy $\beta$, the subtree $T_i$ can be completed so that every player in $T_i$ is playing a best response, with stubborn agents using their fixed strategies.

For a leaf vertex $i$, this invariant follows directly from the best-response
condition. If $i\in S$, then $i$ must use its fixed strategy $s_i$. If
$i\in C$, then $i$ has only one neighbour, namely its parent, and hence it is
a best response if and only if it matches the parent strategy $\alpha$. If
$i\in R$, then it is a best response if and only if it chooses a strategy
different from $\alpha$. Therefore the initialization of $D_i(\alpha)$ is
correct.

Now suppose the invariant holds for all children of an internal non-root
vertex $i$. Fix the parent strategy $\alpha\in A$ and a candidate strategy
$\beta\in A$ for $i$. For each child $j\in\operatorname{ch}(i)$, the set
$D_j(\beta)$ contains exactly the strategies of $j$ that can be extended to a
feasible profile on the subtree $T_j$, conditional on $i$ choosing $\beta$.
Since the child subtrees are disjoint in a tree, choosing one $\gamma_j\in D_j(\beta)$ for each child $j$ is exactly the condition that all child subtrees can be completed consistently.

Given such a choice, the number of neighbours of $i$ choosing strategy $a$ is
$$
q_a=
\mathbf{1}_{\{\alpha=a\}} + \left|\{j\in\operatorname{ch}(i):\gamma_j=a\}\right|, \qquad a\in A.
$$
Thus the algorithm accepts the candidate strategy $\beta$ precisely when
player $i$ itself satisfies its best-response condition: a stubborn agent
uses its fixed strategy, a conformist chooses a strategy with maximal
neighbour count, and a rebel chooses a strategy with minimal neighbour count.
Therefore the invariant also holds for vertex $i$.

By induction, the invariant holds for every non-root vertex. It remains to
consider the root $\rho$. Since the root has no parent, its neighbour counts
are determined only by the strategies of its children:
$$
q_a=\left|\{j\in\operatorname{ch}(\rho):\gamma_j=a\}\right|.
$$
The algorithm checks whether there exists a root strategy $\beta$ and feasible
child strategies $\gamma_j\in D_j(\beta)$ such that the root also satisfies
its best-response condition. By the invariant, each feasible child strategy
can be extended to a feasible profile on the corresponding child subtree.
Hence a feasible root strategy exists if and only if the whole tree admits a
strategy profile in which every player is playing a best response.

This is exactly the definition of a PNE. Therefore the algorithm returns that
a PNE exists if and only if the $m$-strategy CRS game on the tree admits a
PNE.
\end{proof}

A direct implementation enumerates, for each vertex, each possible parent
strategy, each candidate strategy, and all feasible choices of child
strategies. Let $d_i^+=|\operatorname{ch}(i)|$. Since
$d_i^+\leq \Delta-1$ for every non-root vertex, the number of child-strategy
combinations is at most $m^{\Delta-1}$. For each such combination, computing
the neighbour-strategy counts and checking the best-response condition takes
$O(\Delta+m)$ time. Hence the total running time is bounded by $O\bigl(n(\Delta+m)m^{\Delta+1}\bigr).$ In particular, when $m$ and $\Delta$ are fixed constants, the algorithm is linear in the number of players $n$. 


\end{document}